\documentclass{article}
\usepackage{ijcai18}

\usepackage{times}
\usepackage{xcolor}
\usepackage{soul}
\usepackage[utf8]{inputenc}
\usepackage[small]{caption}
\usepackage{multirow}
\usepackage{graphicx}
\usepackage{tabularx}
\usepackage{subfigure}
\usepackage{amssymb,amsmath,amsthm,amsfonts}
\usepackage{algorithm}
\usepackage{algorithmic}
\usepackage{epstopdf}
\newtheorem{theorem}{Theorem}

\usepackage{array}

\title{Adaptive Collaborative Similarity Learning for Unsupervised \\ Multi-view Feature Selection}

\author{
Xiao Dong$^1$,
Lei Zhu$^1$\thanks{Corresponding Author},
Xuemeng Song$^2$,
Jingjing Li$^3$,
Zhiyong Cheng$^4$
\\
$^1$ School of Information Science and Engineering, Shandong Normal University, China \\
$^2$ School of Computer Science and Technology, Shandong University, China \\
$^3$ University of Electronic Science and Technology of China, China \\
$^4$ School of Computing, National University of Singapore, Singapore\\
leizhu0608@gmail.com
}

\begin{document}

\maketitle

\begin{abstract}
In this paper, we investigate the research problem of unsupervised multi-view feature selection. Conventional solutions first simply combine multiple pre-constructed view-specific similarity structures into a collaborative similarity structure, and then perform the subsequent feature selection. These two processes are separate and independent. The collaborative similarity structure remains fixed during feature selection. Further, the simple undirected view combination may adversely reduce the reliability of the ultimate similarity structure for feature selection, as the view-specific similarity structures generally involve noises and outlying entries. To alleviate these problems, we propose an adaptive collaborative similarity learning (ACSL) for multi-view feature selection. We propose to dynamically learn the collaborative similarity structure, and further integrate it with the ultimate feature selection into a unified framework. Moreover, a reasonable rank constraint is devised to adaptively learn an ideal collaborative similarity structure with proper similarity combination weights and desirable neighbor assignment, both of which could positively facilitate the feature selection. An effective solution guaranteed with the proved convergence is derived to iteratively tackle the formulated optimization problem. Experiments demonstrate the superiority of the proposed approach.
\end{abstract}

\section{Introduction}
\label{sec:intro}
With the advent of big data, multi-view features with high dimensions are widely employed to represent the complex data in various research fields, such as multimedia computing, machine learning and data mining \cite{liuananone,liuanantwo,zhutkde,DBLP:journals/tcyb/ZhuSJZX15,CHENG201613,chengmmsys}. On the one hand, with multi-view features, the data could be characterized more precisely and comprehensively from different perspectives. On the other hand, high-dimensional multi-view features will inevitably generate expensive computation cost and cause massive storage cost. Moreover, they may contain adverse noises, outlying entries, irrelevant and correlated features, which may be detrimental to the subsequent learning process \cite{zhuijcai,7605530,7984879}. Unsupervised multi-view feature selection \cite{WANG2016691,DBLP:journals/expert/LiL17} is devised to alleviate the problem. It selects a compact subset of informative features from the original features by dropping irrelevant and redundant features with advanced unsupervised learning. Due to the independence on semantic labels, high computing efficiency and well interpretation capability, unsupervised multi-view feature selection has received considerable attention in literature. It becomes a prerequisite component in various machine learning models \cite{DBLP:journals/kais/LiLL17}.

The key problem of multi-view feature selection is how to effectively exploit the diversity and consistency of multi-view features to collaboratively identify the feature dimensions, which could retain the key characteristics of the original features. Existing approaches can be categorized into two major families. The first kind of methods first concatenates multi-view features into a vector and then directly imports it into the conventional single-view feature selection model. The candidate features are generally ranked based on spectral graph theory. Typical methods of this kind include Laplacian Score (LapScor) \cite{DBLP:conf/nips/HeCN05}, spectral feature selection (SPEC) \cite{DBLP:conf/icml/ZhaoL07} and minimum redundancy spectral feature selection (MRSF) \cite{DBLP:conf/aaai/ZhaoWL10}. Commonly, the pipeline of these methods follows two separate processes: 1) Similarity structure is constructed with fixed graph parameters to describe the geometric structure of data. 2) Sparsity and manifold regularization are employed together to identify the most salient features. Although these methods are reported to achieve certain success, they treat features from different views independently and unfortunately neglect the important view correlations.

Another family of methods considers view correlation when performing feature selection. Representative works include adaptive multi-view feature selection (AMFS) \cite{WANG2016691}, multi-view feature selection (MVFS) \cite{DBLP:conf/sdm/GaoHLT13} and adaptive unsupervised multi-view feature selection (AUMFS) \cite{Feng2013}. These methods first construct multiple view-specific similarity structures\footnote{In this paper, view-specific similarity structure is constructed with the corresponding view-specific feature.} and then perform the subsequent feature selection based on the collaborative (combined) similarity structure. These two processes are separate and independent. The collaborative similarity structure remains fixed during feature selection. The latently involved data noises and outlying entries in the view-specific similarity structures will adversely reduce the reliability of the ultimate collaborative similarity structure for feature selection. Furthermore, conventional approaches generally employ $k$-nearest neighbors assignment to construct the view-specific similarity structures and the simple weighted combination for ultimate similarity structure generation. This strategy can hardly achieve the ideal state for clustering that the number of connected components in the ultimate similarity structure is equal to the number of clusters \cite{DBLP:conf/kdd/NieWH14}. Thus, suboptimal performance may be caused under such circumstance.

In this paper, we introduce an adaptive collaborative similarity learning (ACSL) for unsupervised multi-view feature selection. The main contributions of this paper can be summarized as follows:
\begin{itemize}
  \item  Different from existing solutions, we integrate the collaborative similarity structure learning and multi-view feature selection into a unified framework. The collaborative similarity structure and similarity combination weights could be learned adaptively by considering the ultimate feature selection performance. Simultaneously, the feature selection can preserve the dynamically adjusted similarity structure.\vspace{1mm}
  \item We impose a reasonable rank constraint to adaptively learn an ideal collaborative similarity structure with proper neighbor assignment which could positively facilitate the ultimate feature selection. An effective alternate optimization approach guaranteed with convergence is derived to iteratively solve the formulated optimization problem.\vspace{1mm}
\end{itemize}

\section{Related Work}
\label{sec:rel}
One kind of unsupervised multi-view feature selection methods directly imports the concatenated features in multiple views into the single-view feature selection model. In \cite{DBLP:conf/nips/HeCN05}, Laplacian score (LapScor) is employed to measure the capability of each feature dimension on preserving sample similarity. \cite{DBLP:conf/icml/ZhaoL07} proposes a general spectral theory based learning framework to unify the unsupervised and supervised feature selection. \cite{DBLP:conf/aaai/ZhaoWL10} adopts an embedding model to  handle feature redundancy in the spectral feature selection. These methods generally rank the candidate feature dimensions with various graphs which characterize the manifold structure. They treat features from different views independently and unfortunately ignore the important correlation of different feature views. Another kind of methods directly tackles the multi-view feature selection. They consider view correlations when performing feature selection. Adaptive multi-view feature selection (AMFS) \cite{WANG2016691} is an unsupervised feature selection approach which is developed for human motion retrieval. It describes the local geometric structure of data in each view with local descriptor and performs the feature selection in a general trace ratio optimization. In this method, the feature dimensions are determined with trace ratio criteria. Adaptive unsupervised multi-view feature selection (AUMFS) \cite{Feng2013} addresses the feature selection problem for visual concept recognition. It employs $l_{2,1}$ norm \cite{DBLP:conf/nips/NieHCD10} based sparse regression model to automatically identify discriminative features. In AUMFS, data cluster structure, data similarity and the correlations of different views are considered for feature selection. Multi-view feature selection (MVFS) \cite{DBLP:conf/sdm/GaoHLT13} investigates the feature selection for multi-view data in social media. A learning framework is devised to exploit the relations of views and help each view select relevant features.


\section{The Proposed Methodology}
\label{sec:prop}

\subsection{Notations and Definitions}
Throughout the paper, all the matrices are written in uppercase with boldface. For a matrix $\textbf{M}\in \mathcal{R}^{N\times d}$, its $i_{th}$ row is denoted by $\textbf{M}_{i'}\in \mathcal{R}^{N\times 1}$, its $j_{th}$ column is denoted by $\textbf{M}_{j}\in \mathcal{R}^{d\times 1}$. The element in the $i_{th}$ row and $j_{th}$ column is represented as $M_{ij}$. The trace of the matrix $\textbf{M}$ is denoted as $Tr(\textbf{M})$. The transpose of matrix $\textbf{M}$ is denoted as $\textbf{M}^{\texttt{T}}$. The $l_{2,1}$ norm of the matrix $\textbf{M}$ is denoted as $||\textbf{M}||_{2,1}$, which is calculated by $\sum_{i=1}^{N} \sqrt{\sum_{j=1}^{d} M_{i,j}^2}$.
The Frobenius norm of $\textbf{M}$ is denoted by $||\textbf{M}||_F=\sqrt{\sum_{i=1}^N \sum_{j=1}^d M_{i,j}^2}$. $\textbf{1}$ denotes a column vector whose all elements are one. $\textbf{I}_{k\times k}$ denotes $k\times k$ identify matrix. 

The feature matrix of data in the $v_{th}$ view is denoted as $\textbf{X}^v=[\textbf{x}_1^v, \textbf{x}_2^v, ..., \textbf{x}_N^v]^\texttt{T}\in \mathcal{R}^{N \times d_v}$, $\textbf{x}_1^v\in \mathcal{R}^{d_v \times 1}$, $d_v$ is the dimension of feature in the $v_{th}$ view, $N$ is the number of data samples. We pack the feature matrices in $V$ views $\{\textbf{X}^v\}_{v=1}^V$ and the overall feature matrix of data can be represented as $\textbf{X}=[\textbf{X}^1,\textbf{X}^2,...,\textbf{X}^V]\in \mathcal{R}^{N \times d}$, $\sum_{v=1}^V d_v = d$. The objective of unsupervised multi-view feature selection is to identify $l$ most valuable features with only $\textbf{X}$.

\subsection{Formulation}
The importance of feature dimensions are primarily determined by measuring the their capabilities on preserving the similarity structures in multiple views. In this paper, we develop a unified learning framework to learn an adaptive collaborative similarity structure with automatic neighbor assignment for multi-view feature selection. In our model, the neighbors in the collaborative similarity structure could be adaptively assigned by considering the feature selection performance, and simultaneously the feature selection could preserve the dynamically constructed collaborative similarity structure.
Given $V$ similarity structures constructed in multiple views $\{\textbf{S}^v\}_{v=1}^V$, $V$ is the number of views, we can automatically learn a collaborative similarity structure $\textbf{S}$ by combining $\{\textbf{S}^v\}_{v=1}^V$ with $V$ weights.
\begin{equation}
\small
\begin{aligned}
\label{eq:sh}
&\arg \min_{\textbf{S},\textbf{W}} \ \sum_{j=1}^N ||\textbf{S}_j- \sum_{v=1}^V w_j^v\textbf{S}_j^v||_F^2\\
& s.t. \ \forall j \ \textbf{1}_N^\texttt{T}\textbf{S}_j=1,  \textbf{S}_j \geq \textbf{0},  \textbf{W}_j^\texttt{T}\textbf{1}_V=1
\end{aligned}
\end{equation}
\noindent where $\textbf{S}_j\in \mathcal{R}^{N\times 1}$ characterizes the similarities between any data points with $j$, it should be subjected to the constraint that $\textbf{1}^\texttt{T}\textbf{S}_j=1, \textbf{S}_j \geq \textbf{0}$, $\textbf{W}_j=[w_j^1,w_j^2,...,w_j^V]^\texttt{T}\in \mathcal{R}^{V\times 1}$ is comprised of view weights for the $j_{th}$ column of similarities, it is constrained with $\textbf{W}_j^\texttt{T}\textbf{1}_V=1$, $\textbf{W}=[\textbf{W}_1, \textbf{W}_2,...,\textbf{W}_N]\in \mathcal{R}^{V\times N}$ is view weight matrix for all columns in the similarity structures. As indicated in recent work \cite{DBLP:conf/kdd/NieWH14}, a theoretically ideal similarity structure for clustering should have the property that the number of connected components is equal to the number of clusters. The similarity structure with such neighbor assignment could benefit the subsequent feature selection. Unfortunately, the similarity structure learned from Eq.(\ref{eq:sh}) does not have such desirable property.

To tackle the problem, in this paper, we impose a reasonable rank constraint on the Laplacian matrix of the collaborative similarity structure to enable it to have such property. Our idea is motivated by the following spectral graph theory.
\begin{theorem}
\label{theorem1}
 If the similarity structure $\textbf{S}$ are nonnegative, the multiplicity of eigen-values $0$ corresponding to its Laplacain matrix is equal to the number of components of $\textbf{S}$. \cite{alavi1991graph}
\end{theorem}
As mentioned above, the data points can be directly partitioned into $k$ clusters if the number of components in the similarity structure $\textbf{S}$ is exactly equal to $k$. \textbf{Theorem} \ref{theorem1} indicates that this condition can be achieved if the rank of Laplacian matrix is equal to $n-k$. With the analysis, we add a reasonable rank constraint in Eq.(\ref{eq:sh}) to achieve the condition. The optimization problem becomes
\begin{equation}
\small
\begin{aligned}
\label{eq:si}
&\arg \min_{\textbf{S}, \textbf{W}} \ \sum_{j=1}^N||\textbf{S}_j- \sum_{v=1}^V \ w_j^v \textbf{S}_j^v||_F^2 \\
& s.t. \ \forall j \quad \textbf{1}_N^\texttt{T}\textbf{S}_j=1, \textbf{S}_j \geq \textbf{0}, \textbf{W}_j^\texttt{T}\textbf{1}_V=1, rank(\textbf{L}_S)=n-k\\
\end{aligned}
\end{equation}
\noindent where $\textbf{L}_S=\textbf{D}_S-\frac{\textbf{S}^\texttt{T}+\textbf{S}}{2}$ is the Laplacain matrix of similarity structure $\textbf{S}$, $\textbf{D}_S=\textbf{S}\textbf{1}$ is diagonal matrix. As shown in Eq.(\ref{eq:si}), directly imposing the rank constraint $rank(\textbf{L}_S)=n-k$ will make the above problem hard to solve. Fortunately, according to Ky Fan's Theorem \cite{kafan}, we can have $\sum_{i=1}^k \delta_i(\textbf{L}_S) = \arg\min_{\textbf{F}\in \mathcal{R}^{n\times k}, \textbf{F}^\texttt{T}\textbf{F}=\textbf{I}_k} Tr(\textbf{F}^\texttt{T}\textbf{L}_S\textbf{F})$,
where $\delta_i(\textbf{L}_S)$ is the $i_{th}$ smallest eigen-values of $\textbf{L}_S$ and $\textbf{F}\in \mathcal{R}^{N\times k}$ is the relaxed cluster indicator matrix. Obviously, the rank constraint $rank(\textbf{L}_S)=n-k$ can be satisfied when $\sum_{i=1}^k \delta_i(\textbf{L}_S)=0$. To this end, we reformulate the Eq.(\ref{eq:si}) as the following
simple equivalent form
\begin{equation}
\small
\begin{aligned}
\small
\label{eq:sj}
&\arg\min_{\textbf{F}, \textbf{S}, \textbf{W}} \ \sum_{j=1}^N||\textbf{S}_j- \sum_{v=1}^V w_j^v \textbf{S}_j^v||_F^2 +  \alpha Tr(\textbf{F}^\texttt{T}\textbf{L}_S\textbf{F})\\
& s.t. \ \forall j \quad \textbf{1}_N^\texttt{T}\textbf{S}_j=1, \textbf{S}_j \geq \textbf{0}, \textbf{W}_j^\texttt{T}\textbf{1}_V=1, \textbf{F}\in \mathcal{R}^{N\times k}, \textbf{F}^\texttt{T}\textbf{F}=\textbf{I}_k
\end{aligned}
\end{equation}
\noindent As shown in the above equation, when $\alpha>0$ is large enough, the term $Tr(\textbf{F}^\texttt{T}\textbf{L}_S\textbf{F})$ is forced to be infinitely approximate 0 and the rank constraint can be satisfied accordingly. By simply transforming the rank constraint to trace in objective function, the problem in Eq.(\ref{eq:si}) can be tackled more easily.

The selected features should preserve the dynamically learned similarity structure. Conventional approaches separate the similarity structure construction and feature selection into two independent processes, which will potentially lead to sub-optimal performance. In this paper, we learn the collaborative similarity structure dynamically and further integrate it with feature selection into a unified framework. Specifically, based on the collaborative similarity structure learning in Eq.(\ref{eq:sj}), we employ sparse regression model to learn a projection matrix $\textbf{P}\in \mathcal R^{d\times k}$, so that the projected low-dimensional data $\textbf{X}\textbf{P}$ can approximate the relaxed cluster indicator $\textbf{F}$. To select the features, we impose $l_{2,1}$ norm penalty on $\textbf{P}$ to force it with row sparsity. The importance of features can be measured by the $l_2$ norm of each row feature in $\textbf{P}$. The overall optimization formulation can be derived as
\begin{equation}
\small
\begin{aligned}
\label{eq:overall}
& \arg \min_{\textbf{P}, \textbf{F}, \textbf{S}, \textbf{W}} \ \boldsymbol{\Omega}(\textbf{P}, \textbf{F}, \textbf{S}, \textbf{W})= \sum_{j=1}^N||\textbf{S}_j- \sum_{v=1}^V w_j^v \textbf{S}_j^v||_F^2 +  \\
& \qquad \qquad \qquad \alpha Tr(\textbf{F}^\texttt{T}\textbf{L}_S\textbf{F})+ \beta(||\textbf{X}\textbf{P}-\textbf{F}||_F^2 + \gamma||\textbf{P}||_{2,1})\\
& s.t. \ \forall j \quad \textbf{1}_N^\texttt{T}\textbf{S}_j=1, \textbf{S}_j \geq \textbf{0}, \textbf{W}_j^\texttt{T}\textbf{1}_V=1, \textbf{F}\in \mathcal{R}^{N\times k}, \textbf{F}^\texttt{T}\textbf{F}=\textbf{I}_k
\end{aligned}
\end{equation}
With $\textbf{P}$, the importance of features are measured by $||\textbf{P}_{i'}||_2$. The features with the $l$ largest values can be finally determined.

\subsection{Alternate Optimization}
As shown in Eq.(\ref{eq:overall}), the objective function is not convex to three variables simultaneously. In this paper, we propose an effective alternate optimization to iteratively solve the problem. Specifically, we optimize one variable by fixing the others.

\textbf{Update $\textbf{P}$}. By fixing the other variables, the optimization for $\textbf{P}$ can be derived as
\begin{equation}
\small
\begin{aligned}
\label{eq:updateP}
\arg\min_{\textbf{P}} & \ ||\textbf{X}\textbf{P}-\textbf{F}||_F^2 + \gamma||\textbf{P}||_{2,1}\\
\end{aligned}
\end{equation}
This equation is not differentiable. Hence, we transform it to following equivalent equation  \cite{DBLP:conf/nips/NieHCD10}
\begin{equation}
\small
\begin{aligned}
\label{eq:updateP1}
\arg\min_{\textbf{P}} & \ ||\textbf{X}\textbf{P}-\textbf{F}||_F^2+ \gamma Tr(\textbf{P}^\texttt{T}\boldsymbol{\Gamma} \textbf{P})
\end{aligned}
\end{equation}
\noindent $\boldsymbol{\Gamma}\in \mathcal{R}^{d\times d}$ is diagonal matrix whose $i_{th}$ diagonal element is $\frac{1}{2\sqrt{\textbf{P}_{i'}\textbf{P}_{i'}^\texttt{T}}+\epsilon}$.
$\epsilon$ is small enough constant. It is used to avoid the condition that $||\textbf{P}_{i'}||_2$ is zero. By calculating the derivations of the objective function with $\textbf{P}$ and setting it to zeros, we can obtain the updating rule for $\textbf{P}$ as
\begin{equation}
\small
\begin{aligned}
\label{eq:updateW}
\textbf{P} = (\textbf{X}^\texttt{T}\textbf{X}+\gamma \boldsymbol{\Gamma})^{-1}\textbf{X}^\texttt{T}\textbf{F}
\end{aligned}
\end{equation}
Note that $\boldsymbol{\Gamma}$ is dependent on $\textbf{P}$. We develop an iterative approach to solve $\textbf{P}$ and $\boldsymbol{\Gamma}$ until convergence. Specifically, we fix $\boldsymbol{\Gamma}$ to solve $\textbf{P}$, and vice versa.

\textbf{Update $\textbf{F}$}. By fixing the other variables, the optimization for $\textbf{F}$ can be derived as
\begin{equation}
\small
\begin{aligned}
\label{eq:updateF}
\arg\min_{\textbf{F}} & \ Tr(\textbf{F}^\texttt{T}\textbf{L}_S\textbf{F}) + \beta(||\textbf{X}\textbf{P}-\textbf{F}||_F^2 + \gamma Tr(\textbf{P}^\texttt{T}\boldsymbol{\Gamma} \textbf{P})) s.t. \textbf{F}^\texttt{T}\textbf{F}=\textbf{I}_c\\
\end{aligned}
\end{equation}
By substituting Eq.(\ref{eq:updateW}) into the objective function in Eq.(\ref{eq:updateF}), we arrive at
\begin{equation}
\small
\begin{aligned}
& Tr(\textbf{F}^\texttt{T}\textbf{L}_S\textbf{F}) + \beta(||\textbf{X}\textbf{P}-\textbf{F}||_F^2 + \gamma Tr(\textbf{P}^\texttt{T}\boldsymbol{\Gamma} \textbf{P})) \\
& = Tr(\textbf{F}^\texttt{T} (\textbf{L}_S+\beta \textbf{I}_N-\beta \textbf{X}\textbf{Q}^{-1}\textbf{X}^\texttt{T}) \textbf{F})
\end{aligned}
\end{equation}
\noindent where $\textbf{Q}=\textbf{X}^\texttt{T}\textbf{X}+\gamma \boldsymbol{\Gamma}$. With the transformation, the optimization for updating $\textbf{F}$ can be solved by simple eigen-decomposition on the matrix $\textbf{L}_S+\beta \textbf{I}_N-\beta \textbf{X}\textbf{Q}^{-1}\textbf{X}^\texttt{T}$. Specifically, the columns of $\textbf{F}$ are comprised of the $k$ eigenvectors corresponding to the $k$ smallest eigenvalues.

\textbf{Update $\textbf{S}$}. By fixing the other variables, the optimization for $\textbf{S}$ becomes
\begin{equation}
\small
\begin{aligned}
\label{eq:updateS}
& \arg\min_{\textbf{S}} \sum_{j=1}^N||\textbf{S}_j- \sum_{v=1}^V w_j^v \textbf{S}_j^v||_F^2 +  \alpha Tr(\textbf{F}^\texttt{T}\textbf{L}_S\textbf{F}) \\
& s.t. \ \forall j \ \textbf{1}_N^\texttt{T}\textbf{S}_j=1, \textbf{S}_j \geq \textbf{0}\\
\end{aligned}
\end{equation}
The above equation can be rewritten as
\begin{equation}
\small
\begin{aligned}
\label{eq:updateS1}
& \arg\min_{\textbf{S}_j} \sum_{j=1}^N||\textbf{S}_j- \sum_{v=1}^V w_j^v \textbf{S}^v_j||_F^2 +  \alpha \sum_{i,j=1}^N ||\textbf{f}_i-\textbf{f}_j||_2^2 S_{i,j}\\
& s.t. \ \forall j \ \textbf{1}_N^\texttt{T}\textbf{S}_j=1, \textbf{S}_j \geq \textbf{0}\\
\end{aligned}
\end{equation}
\noindent where $S_{i,j}$ denotes the element in the $i_{th}$ row and $j_{th}$ column of $\textbf{S}$. The optimization processes for the columns of $\textbf{S}$ are independent with each other. Hence, they can be optimized separately. Formally, $\textbf{S}$ can be solved by
\begin{equation}
\small
\begin{aligned}
\label{eq:updateS2}
& \arg\min_{\textbf{S}_j} ||\textbf{S}_j- \sum_{v=1}^V w_j^v \textbf{S}^v_j||_F^2 +  \alpha \textbf{A}_j^\texttt{T}\textbf{S}_j \quad s.t. \ \forall j \ \textbf{1}_N^\texttt{T}\textbf{S}_j=1, \textbf{S}_j \geq \textbf{0}\\
\end{aligned}
\end{equation}
Let $\textbf{A}_{j}$ be row vector with $N\times 1$ dimensions. Its $i_{th}$ element is $||\textbf{f}_i-\textbf{f}_j||_2^2$. The above optimization formula can be transformed as
\begin{equation}
\small
\begin{aligned}
\label{eq:updateS3}
& \arg\min_{\textbf{S}_j} ||\textbf{S}_j + \frac{\alpha}{2}\textbf{A}_{j}- \sum_{v=1}^V w_j^v \textbf{S}^v_j||_F^2 \quad s.t. \ \forall j \ \textbf{1}_N^\texttt{T}\textbf{S}_j=1, \textbf{S}_j \geq \textbf{0}\\
\end{aligned}
\end{equation}
This problem can be solved by an efficient iterative algorithm \cite{DBLP:conf/ijcai/HuangNH15}.

\textbf{Update $\textbf{W}$}. Similar to $\textbf{S}$, the optimization processes for the columns of $\textbf{W}$ are independent with each other. Hence, they can be optimized separately. Formally, its $j_{th}$ column $\textbf{W}_j$ is solved by
\begin{equation}
\small
\begin{aligned}
\label{eq:updateweight}
& \arg\min_{\textbf{W}_j} \ ||\textbf{S}_j - \sum_{v=1}^V w_j^v \textbf{S}^v_j||_F^2 \quad s.t. \ \textbf{W}_j^\texttt{T}\textbf{1}_V=1 \\
\end{aligned}
\end{equation}
The objective function in Eq.(\ref{eq:updateweight}) can be rewritten as
\begin{equation}
\small
\begin{aligned}
&||\textbf{S}_j - \sum_{v=1}^V w_j^v \textbf{S}^v_j||_F^2=||\sum_{v=1}^V w_j^v\textbf{S}_j - \sum_{v=1}^V w_j^v \textbf{S}^v_j||_F^2 \\
& = ||\sum_{v=1}^V w_j^v(\textbf{S}_j - \textbf{S}^v_j)||_F^2 = ||\textbf{B}_j\textbf{W}_j||_F^2 = \textbf{W}_j^\texttt{T}\textbf{B}_j^\texttt{T}\textbf{B}_j\textbf{W}_j
\end{aligned}
\end{equation}
\noindent where $\textbf{B}_j^v=\textbf{S}_j - \textbf{S}^v_j$, $\textbf{B}_j=[\textbf{B}_j^1,...,\textbf{B}_j^v,..,\textbf{B}_j^V]$.

We can obtain the Lagrangian function of problem (\ref{eq:updateweight})
\begin{equation}
\small
\begin{aligned}
\label{eq:updateweight2}
\mathcal{L}(W_j, \psi)= \textbf{W}_j^\texttt{T}\textbf{B}_j^\texttt{T}\textbf{B}_j\textbf{W}_j + \psi(1-\textbf{W}_j^\texttt{T} \textbf{1}_V)
\end{aligned}
\end{equation}
\noindent $\psi$ is also Lagrangian multiplier. By calculating the derivative of (\ref{eq:updateweight2}) with $\textbf{W}_j$ and setting it to 0, we obtain the updating rule of $\textbf{W}_j$ as
\begin{equation}
\small
\begin{aligned}
\label{eq:updatefeature}
\textbf{W}_j = \frac{(\textbf{B}_j^\texttt{T}\textbf{B}_j)^{-1}\textbf{1}_V}{\textbf{1}_V^\texttt{T}(\textbf{B}_j^\texttt{T}\textbf{B}_j)^{-1}\textbf{1}_V}
\end{aligned}
\end{equation}

\begin{table*}
\scriptsize
\centering
\vspace{-2mm}
\begin{tabular}{ccccccccc}
\hline
\multicolumn{1}{l}{Dataset}             & Feature dimension & LapScor             & SPEC              & MRSF              & MVFS              & AUMFS              & AMFS               & ACSL               \\ \hline
                                        & 100               & 0.2867 & \textbf{0.2952}  & 0.2838 & 0.2762 & 0.2810  & 0.28571 & \textbf{0.3000}  \\
                                        & 200               & 0.2952  & 0.2905 & \textbf{0.3152} & 0.2905 & \textbf{0.3143}  & 0.2895  & 0.3124  \\
\multicolumn{1}{l}{\textbf{MSRC-V1}}   & 300               & 0.2905   & \textbf{0.3119} & 0.2895 & 0.2833 & 0.2833  & 0.2952  & \textbf{0.3124}  \\
                                        & 400               & 0.2952   & \textbf{0.3181} & 0.3057 & 0.3000 & 0.2952  & 0.2924  & \textbf{0.3219}  \\
                                        & 500               & 0.3038   & 0.2976 & 0.3038 & \textbf{0.3095} & 0.3048  & 0.2990  & \textbf{0.3400}  \\ \hline

                                        & 100               & 0.5844   & 0.4795 & \textbf{0.6207} & 0.5938 & 0.3345  & 0.3302 & \textbf{0.6106}  \\
                                        & 200               & \textbf{0.6148}   & 0.5520 & 0.6002 & 0.5820 & 0.4225  & 0.4226  & \textbf{0.6389} \\
\multicolumn{1}{l}{\textbf{Handwritten Numeral}} & 300               & \textbf{0.5980}   & 0.5384 & \textbf{0.6028} & 0.5737 & 0.4757  & 0.4497  & 0.5930  \\
                                        & 400               & 0.6068  & \textbf{0.6102} & 0.5890 & 0.5808 & 0.4909  & 0.4755 & \textbf{0.6327} \\
                                        & 500               & \textbf{0.5909}   & 0.5666 & 0.5795 & 0.5888 & 0.4889  & 0.5006 & \textbf{0.5969}  \\ \hline

                                        & 100               & \textbf{0.2873}   & \textbf{0.2873} & 0.2851 & 0.2717 & 0.1305  & 0.2165  & 0.2861  \\
                                        & 200               & \textbf{0.2896}   & 0.2840 & 0.2754 & 0.2774 & 0.1274  & 0.2313  & \textbf{0.2924}  \\
\multicolumn{1}{l}{\textbf{Youtube}}         & 300               & 0.2835  & 0.2832 & \textbf{0.2862} & 0.2828 & 0.1357 & 0.2374  & \textbf{0.2906}  \\
                                        & 400               & 0.2862  & \textbf{0.2889} & 0.2779 & 0.2807 & 0.1329  & 0.2433  & \textbf{0.2993}  \\
                                        & 500               & \textbf{0.2857}  & 0.2853 & 0.2802 & 0.2854 & 0.1329  & 0.2546  & \textbf{0.3003}  \\ \hline

                                       & 100               & 0.3687   & 0.3327 &  0.3707   & 0.2044  & 0.4231   & \textbf{0.4313} & \textbf{0.5845}  \\
                                        & 200             & 0.3619  & 0.3295   & 0.3501   & 0.2104  & 0.4656   & \textbf{0.4816}  &  \textbf{0.5616}  \\
\multicolumn{1}{l}{\textbf{Outdoor Scene}}         & 300            & 0.3634   & 0.3740  &  0.3576  & 0.2150  & \textbf{0.4949}   & 0.4854 & \textbf{0.5801}   \\
                                        & 400              & 0.3804  & 0.3653 & 0.3679  &  0.2153  & \textbf{0.5061}  & 0.4926  &  \textbf{0.5927}   \\
                                        & 500              & 0.3574   &  0.3620  & 0.3687   & 0.2255  & 0.5003   &\textbf{0.5045}  & \textbf{0.6103}
\\ \hline
\end{tabular}
\caption{ACC of different methods with different numbers of selected features by using K-means for clustering.}
\label{Tab1}
\end{table*}

The main steps for solving problem (\ref{eq:overall}) are summarized in Algorithm \ref{alg:summary}.
\begin{algorithm}
\vspace{-1mm}
\caption{Multi-view feature selection via collaborative similarity structure learning with adaptive neighbors.}
\label{alg:summary}
\begin{algorithmic}[1]
\REQUIRE ~~\\
The pre-constructed similarity structures in $v$ views $\{\textbf{S}^v\}_{v=1}^V$, the number of clusters $k$, the parameters $\alpha,\beta,\gamma$.\\
\ENSURE ~~\\
The collaborative similarity structure $\textbf{S}$, the projection matrix $\textbf{P}$ for feature selection, $l$ identified features.\\
\STATE Initialize $\textbf{W}$ with $\frac{1}{V}$, the collaborative similarity structure $\textbf{S}$ with the weighted sum of $\{\textbf{S}^v\}_{v=1}^V$. We also initialize $\textbf{F}$ with the solution of problem (\ref{eq:updateF}) by substituting the Laplacian matrix calculated from the new $\textbf{S}$.
\REPEAT
    \STATE Update $\textbf{P}$ with Eq.(\ref{eq:updateW}).
    \STATE Update $\textbf{F}$ by solving the problem in Eq.(\ref{eq:updateF}).
    \STATE Update $\textbf{S}$ with Eq.(\ref{eq:updateS3}).
    \STATE Update $\textbf{W}$ with Eq.(\ref{eq:updatefeature}).
\UNTIL{Convergence}\\
\textbf{Feature Selection}
\STATE Calculate $||\textbf{P}_{i'}||_2, (i=1,2,...,d)$ and rank them in descending order. The $l$ features with the top rank orders are finally determined as the features to be selected.
\end{algorithmic}
\end{algorithm}
\vspace{-3mm}
\subsection{Convergence Analysis}
\label{conver_analysis}
The convergence of solving problem (\ref{eq:updateP1}) can be proven by the following theorem.
\begin{theorem}
\label{theoremupdatP}
  The iterative optimization process for solving Eq.(\ref{eq:updateP}) will monotonically decrease the objective function value until convergence.
  \vspace{-3mm}
\end{theorem}
\begin{proof}
  Let $\widehat{\textbf{P}}$ be the newly updated $\textbf{P}$, we can obtain the following inequality
\begin{equation}
\small
\begin{aligned}
\label{ieq:1}
& ||\textbf{X}\widehat{\textbf{P}}-\textbf{F}||_F^2+ \gamma Tr(\widehat{\textbf{P}}^\texttt{T}\boldsymbol{\Gamma} \widehat{\textbf{P}}) \leq ||\textbf{X}\textbf{P}-\textbf{F}||_F^2+ \gamma Tr(\textbf{P}^\texttt{T}\boldsymbol{\Gamma} \textbf{P})
\end{aligned}
\end{equation}
By adding $\gamma \sum_{i=1}^d \frac{\epsilon}{2\sqrt{\textbf{P}_{i'}^\texttt{T}\textbf{P}_{i'}}+\epsilon}$ to the both sides of the inequality (\ref{ieq:1}) and substituting $\boldsymbol\Gamma$, the inequality can be rewritten as
\begin{equation}
\small
\begin{aligned}
\label{eq:theorem1}
||\textbf{X}\widehat{\textbf{P}}-\textbf{F}||_F^2+ \gamma \sum_{i=1}^d \frac{\widehat{\textbf{P}}_{i'}^\texttt{T}\widehat{\textbf{P}}_{i'}+\epsilon}{2\sqrt{\textbf{P}_{i'}^\texttt{T}\textbf{P}_{i'}+\epsilon}} \leq  \\ ||\textbf{X}\textbf{P}-\textbf{F}||_F^2+ \gamma \sum_{i=1}^d \frac{\textbf{P}_{i'}^\texttt{T}\textbf{P}_{i'}+\epsilon}{2\sqrt{\textbf{P}_{i'}^\texttt{T}\textbf{P}_{i'}+\epsilon}}
\end{aligned}
\end{equation}
On the other hand, according to the \textbf{Lemma} 1 in \cite{DBLP:conf/nips/NieHCD10}, we can obtain that for any positive number $u$ and $v$, we can have
\begin{equation}
\small
\begin{aligned}
\sqrt{u}-\frac{\sqrt{u}}{2\sqrt{v}}\leq \sqrt{v}-\frac{\sqrt{v}}{2\sqrt{v}}
\end{aligned}
\end{equation}

Then, we can obtain that
\begin{equation}
\small
\begin{aligned}
\label{eq:theorem2}
\gamma \sum_{i=1}^d\sqrt{\widehat{\textbf{P}}_{i'}^\texttt{T}\widehat{\textbf{P}}_{i'} + \epsilon} - \gamma \sum_{i=1}^d \frac{\widehat{\textbf{P}}_{i'}^\texttt{T}\widehat{\textbf{P}}_{i'}+\epsilon}{2\sqrt{\textbf{P}_{i'}^\texttt{T}\textbf{P}_{i'}+\epsilon}}\\
\leq  \gamma \sum_{i=1}^d\sqrt{\textbf{P}_{i'}^\texttt{T}\textbf{P}_{i'}+\epsilon} - \gamma \sum_{i=1}^d \frac{\textbf{P}_{i'}^\texttt{T}\textbf{P}_{i'}+\epsilon}{2\sqrt{\textbf{P}_{i'}^\texttt{T}\textbf{P}_{i'}+\epsilon}}
\end{aligned}
\end{equation}
By summing the above inequalities (\ref{eq:theorem1}) and (\ref{eq:theorem2}), we arrive at
\begin{equation}
\small
\begin{aligned}
&||\textbf{X}\widehat{\textbf{P}}-\textbf{F}||_F^2+ \gamma \sum_{i=1}^d\sqrt{\widehat{\textbf{P}}_{i'}^\texttt{T}\widehat{\textbf{P}}_{i'} + \epsilon}\\
& \leq ||\textbf{X}\textbf{P}-\textbf{F}||_F^2+ \gamma \sum_{i=1}^d\sqrt{\textbf{P}_{i'}^\texttt{T}\textbf{P}_{i'}+\epsilon}
\end{aligned}
\end{equation}
We can derive that
\begin{equation}
\small
\begin{aligned}
||\textbf{X}\widehat{\textbf{P}}-\textbf{F}||_F^2+ \gamma ||\widehat{\textbf{P}}||_{2,1}
\leq ||\textbf{X}\textbf{P}-\textbf{F}||_F^2+  \gamma ||\textbf{P}||_{2,1}
\end{aligned}
\end{equation}
\end{proof}

The convergence of solving Algorithm \ref{alg:summary} can be proven by the following theorem.
\begin{theorem}
  The iterative optimization in Algorithm \ref{alg:summary} can monotonically decrease the objective function of problem (\ref{eq:overall}) until convergence.
\end{theorem}
\begin{proof}
  As shown in \textbf{Theorem} \ref{theoremupdatP}, updating $\textbf{P}$ will monotonically decrease the objective function in problem (\ref{eq:overall}) ($t$ is number of iterations).
  \begin{equation}
  \small
    \boldsymbol{\Omega}(\textbf{P}^{(t)}, \textbf{F}, \textbf{S}, \textbf{W})\geq \boldsymbol{\Omega}(\textbf{P}^{(t+1)}, \textbf{F}, \textbf{S}, \textbf{W})
  \end{equation}

  By fixing other variables and updating $\textbf{F}$, the objective function in Eq.(\ref{eq:updateF}) is convex (The Hessian matrix of the Lagrangian function of Eq.(\ref{eq:updateF}) is positive semidefinite \cite{alavi1991graph}). Therefore, we can obtain that
  \begin{equation}
  \small
    \boldsymbol{\Omega}(\textbf{P}, \textbf{F}^{(t)}, \textbf{S}, \textbf{W})\geq \boldsymbol{\Omega}(\textbf{P}, \textbf{F}^{(t+1)}, \textbf{S}, \textbf{W})
  \end{equation}

  By fixing other variables and updating $\textbf{S}$, optimizing the Eq.(\ref{eq:updateS3}) is a typical Quadratic programming problem. The Hessian matrix of the Lagrangian function of problem (\ref{eq:updateS3}) is also $\textbf{1}_N\textbf{1}_N^\texttt{T}$ that is positive semidefinite.
  Therefore, we can obtain that
  \begin{equation}
  \small
    \boldsymbol{\Omega}(\textbf{P}, \textbf{F}, \textbf{S}^{(t)}, \textbf{W})\geq \boldsymbol{\Omega}(\textbf{P}, \textbf{F}, \textbf{S}^{(t+1)}, \textbf{W})
  \end{equation}

  By fixing other variables and updating $\textbf{W}$, the Hessian matrix of Eq.(\ref{eq:updateweight2}) is $\textbf{B}_j^\texttt{T}\textbf{B}_j$. It is positive semidefinite as
  $\textbf{W}_j^\texttt{T}\textbf{B}_j^\texttt{T}\textbf{B}_j\textbf{W}_j)=||\textbf{B}_j\textbf{W}_j||_F^2\geq 0$. Hence, the objective function for optimizing $\textbf{W}$ is also convex. Then, we arrive at
  \begin{equation}
  \small
    \boldsymbol{\Omega}(\textbf{P}, \textbf{F}, \textbf{S}, \textbf{W}^{(t)})\geq \boldsymbol{\Omega}(\textbf{P}, \textbf{F}, \textbf{S}, \textbf{W}^{(t+1)})
  \end{equation}

\end{proof}
\begin{table*}[]
\scriptsize
\centering
\vspace{-2mm}
\begin{tabular}{ccccccccc}
\hline
\multicolumn{1}{l}{Dataset}             & Feature dimension & LapScor           & SPEC              & MRSF              & MVFS              & AUMFS             & AMFS               & ACSL              \\ \hline
                                        & 100               & \textbf{0.1653}  & \textbf{0.1930}   & 0.1555  & 0.1362  & 0.1146  & 0.12681  & 0.1635  \\
                                        & 200               & 0.1730  & 0.1518  & 0.1754  & 0.1502  & \textbf{0.1799} & 0.1591   & \textbf{0.1875}  \\
\multicolumn{1}{l}{\textbf{MSRC-v1}}   & 300               & 0.1632 & 0.1637 & \textbf{0.1713} & 0.1358 & 0.1341 & 0.1609  & \textbf{0.1912} \\
                                        & 400               & 0.1815 & \textbf{0.2195} & 0.1787 & 0.1407 & 0.1716 & 0.1595   & \textbf{0.1905} \\
                                        & 500               & 0.1672 & \textbf{0.2027} & 0.1813 & 0.1798 & 0.1735 & 0.1670  & \textbf{0.2146} \\ \hline

                                        & 100               & \textbf{0.5967} & 0.4751 & 0.5927 & 0.5485 & 0.2738 & 0.2744  & \textbf{0.6403} \\
                                        & 200               & \textbf{0.6050} & 0.5413 & 0.5943 & 0.5538 & 0.3720 & 0.3718  & \textbf{0.6513} \\
\multicolumn{1}{l}{\textbf{Handwritten Numeral}} & 300               & 0.5962 & \textbf{0.6068} & \textbf{0.6051} & 0.5584 & 0.4101 & 0.4013 & 0.5932 \\
                                        & 400               & 0.6014 & 0.6010 & \textbf{0.6015} & 0.5690 & 0.4436 & 0.4423  & \textbf{0.6025} \\
                                        & 500               & \textbf{0.6078} & 0.5799 & \textbf{0.5983} & 0.5974 & 0.4796 & 0.4831  & 0.5926 \\ \hline

                                        & 100               & \textbf{0.2690} & 0.2683 & 0.2610 & 0.2531 & 0.0121 & 0.1280 & \textbf{0.2705} \\
                                        & 200               & \textbf{0.2693} & 0.2688 & 0.2561 & 0.2604 & 0.0108 & 0.1474 & \textbf{0.2699} \\
\multicolumn{1}{l}{\textbf{Youtube}}         & 300               & 0.2627 & \textbf{0.2673} & \textbf{0.2677} & 0.2605 & 0.0152 & 0.1597  & 0.2570 \\
                                        & 400               & \textbf{0.2670} & 0.2647 & 0.2606 & \textbf{0.2736} & 0.0142 & 0.1817  & \textbf{0.2743} \\
                                        & 500               & 0.2641 & 0.2696 & 0.2635 & \textbf{0.2771} & 0.0123 & 0.1982 & \textbf{0.2736} \\ \hline

										& 100               & 0.2228   & 0.1933  &  0.2203   & 0.0595 & \textbf{0.3314}  & 0.3267  & \textbf{0.4717}  \\
                                        & 200               & 0.2174  & 0.2023  &  0.2021  & 0.0522 & \textbf{0.3801}  & 0.3772 & \textbf{0.4860}   \\
\multicolumn{1}{l}{\textbf{Outdoor Scene}}       & 300               & 0.2264   & 0.2414 &  0.2142   & 0.0562  & \textbf{0.4157}   & 0.3975  & \textbf{0.4838}   \\
                                        & 400               & 0.2358   & 0.2337  & 0.2211 & 0.0588  & \textbf{0.4152}   & 0.4023   & \textbf{0.5111}   \\
                                        & 500               & 0.2163  & 0.2316  &  0.2196   & 0.0769  & \textbf{0.4210}  & 0.4159  & \textbf{0.5211}
\\ \hline
\end{tabular}
\caption{NMI of different methods with different numbers of selected features by using K-means for clustering.}
\label{Tab2}
\vspace{-2mm}
\end{table*}

\section{Experiments}
\label{sec:exp}
\subsection{Experimental Datasets}
1) \textbf{MSRC-v1} \cite{1541329}. The dataset contains 240 images in 8 class as a whole. Following the setting in \cite{Grauman2006Unsupervised}, we select 7 classes composed of tree, building, airplane, cow, face, car, bicycle and each class has 30 images. We extract 5 visual features from each image: color moment with dimension 48, GIST with 512 dimension, SIFT with dimension 1230, CENTRIST feature with 210 dimension, and local binary pattern (LBP) with 256 dimension. 2) \textbf{Handwritten Numeral} \cite{vanbreukelen1998handwritten}. This dataset is comprised of 2,000 data points from 0 to 9 digit classes. 6 features are used to represent each digit. They are 76 dimensional Fourier coefficients of the character shapes,
216 dimensional profile correlations, 64 dimensional Karhunen-love coefficients, 240 dimensional pixel averages in $2\times3$ windows,
47 dimensional Zernike moment and 6 dimensional morphological features. 3) \textbf{Youtube} \cite{5206845}. This real-world dataset is collected from Youtube. It contains intended camera motion, variations of the object scale, viewpoint, illumination
and cluttered background. The dataset is comprised of 1,596 video sequences in 11 actions. 4) \textbf{Outdoor Scene}  \cite{Monadjemi2002Experiments}.  The outdoor scene dataset  contains 2,688 color images that belong to 8 outdoor scene categories. 4 visual features are extracted from each image: color moment with dimension 432, GIST with dimension 512, HOG with dimension 256, and LBP with dimension 48.


\subsection{Experimental Setting}
\textbf{Baselines}. We compare ACSL with several representative unsupervised multi-view feature selection methods on clustering performance. The compared methods include, three single view feature selection approaches (Laplacian score (LapScor) \cite{DBLP:conf/nips/HeCN05}, spectral feature selection (SPEC) \cite{DBLP:conf/icml/ZhaoL07} and minimum redundancy spectral feature selection (MRSF) \cite{DBLP:conf/aaai/ZhaoWL10}), and three multi-view feature selection approach (adaptive multi-view feature selection (AMFS) \cite{WANG2016691}, multi-view feature selection (MVFS) \cite{DBLP:conf/sdm/GaoHLT13} and adaptive unsupervised multi-view feature selection (AUMFS) \cite{Feng2013}). \textbf{Evluation Metrics.}
We employ standard metrics: clustering accuracy (ACC) and normalized mutual information (NMI), for performance comparison. Each experiment is performed 50 times and the mean results are reported. \textbf{Parameter Setting.} In implementation of all methods, the neighbor graph is adopted to construct the initial affinity matrices. The number of neighbors is set to 10 in all methods. In ACSL, $\alpha, \beta, \gamma$ are chosen from $10^{-4}$ to $10^4$. The parameters in all compared approaches are carefully adjusted to report the best results.
\begin{figure*}
\centering                     
\subfigure[$\alpha$ is fixed to $10^4$]{
\includegraphics[width=45mm]{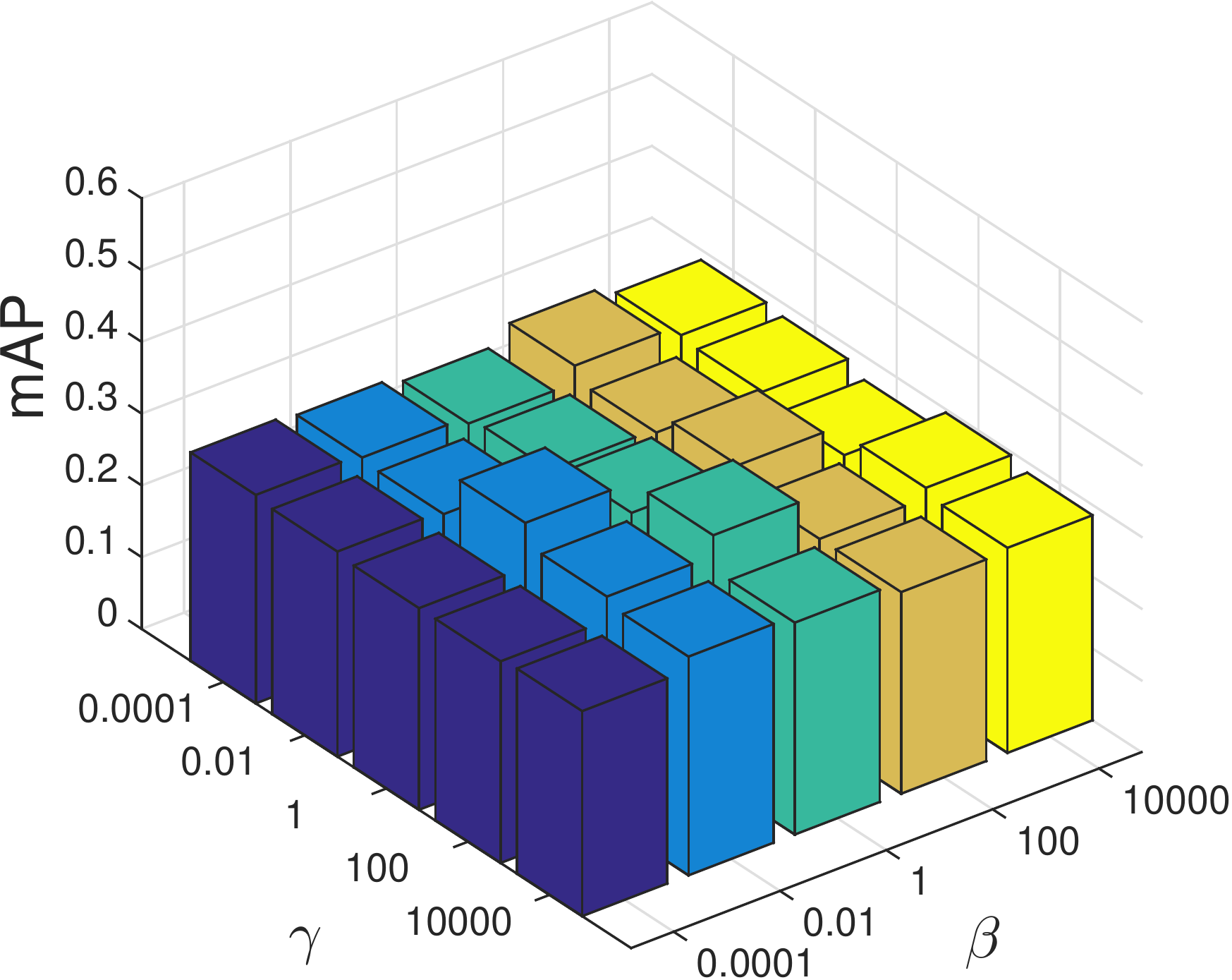}}
\hspace{0.1in}
\subfigure[$\beta$ is fixed to $10^{-2}$]{
\includegraphics[width=45mm]{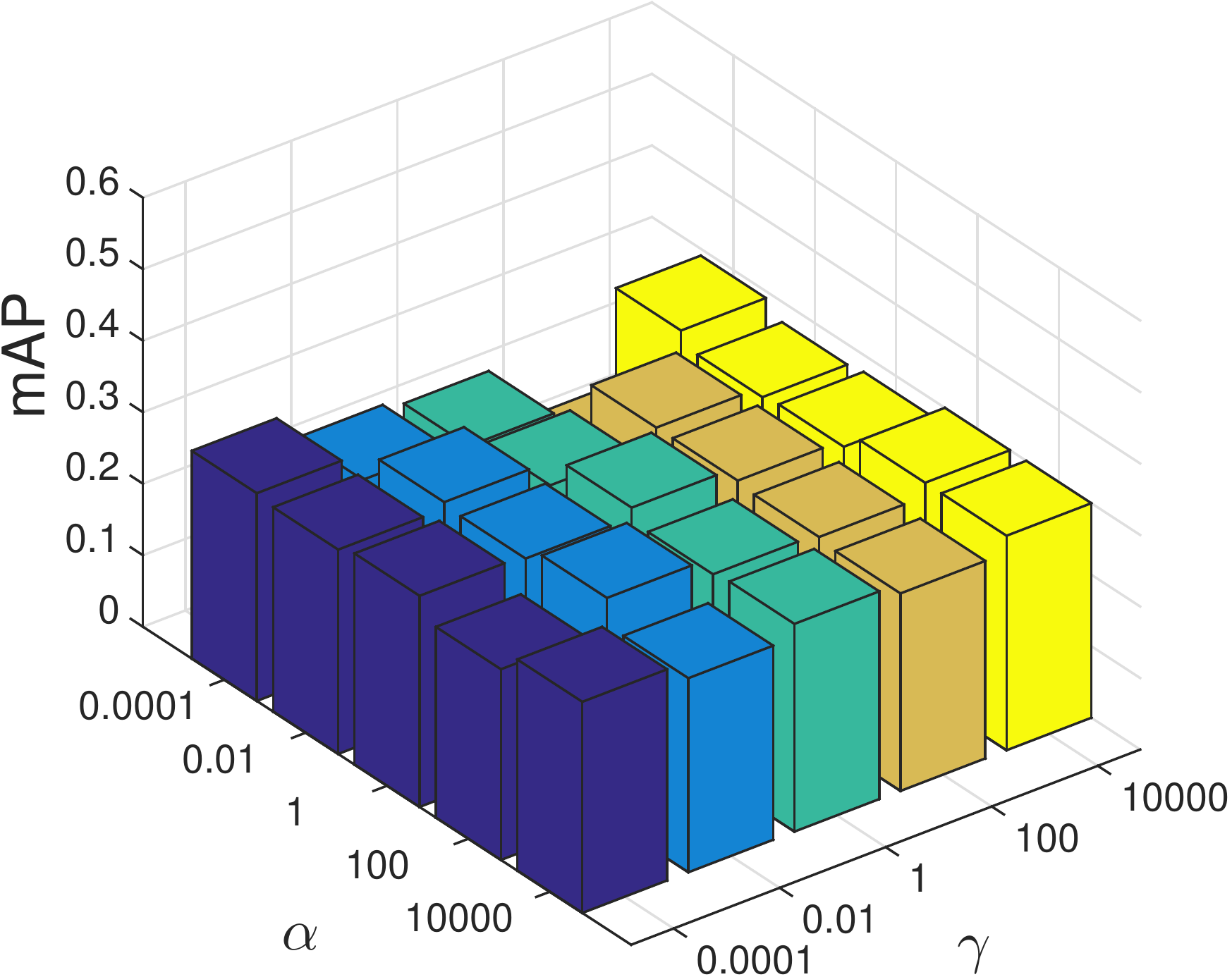}}
\subfigure[$\gamma$ is fixed to $10^3$]{
\includegraphics[width=45mm]{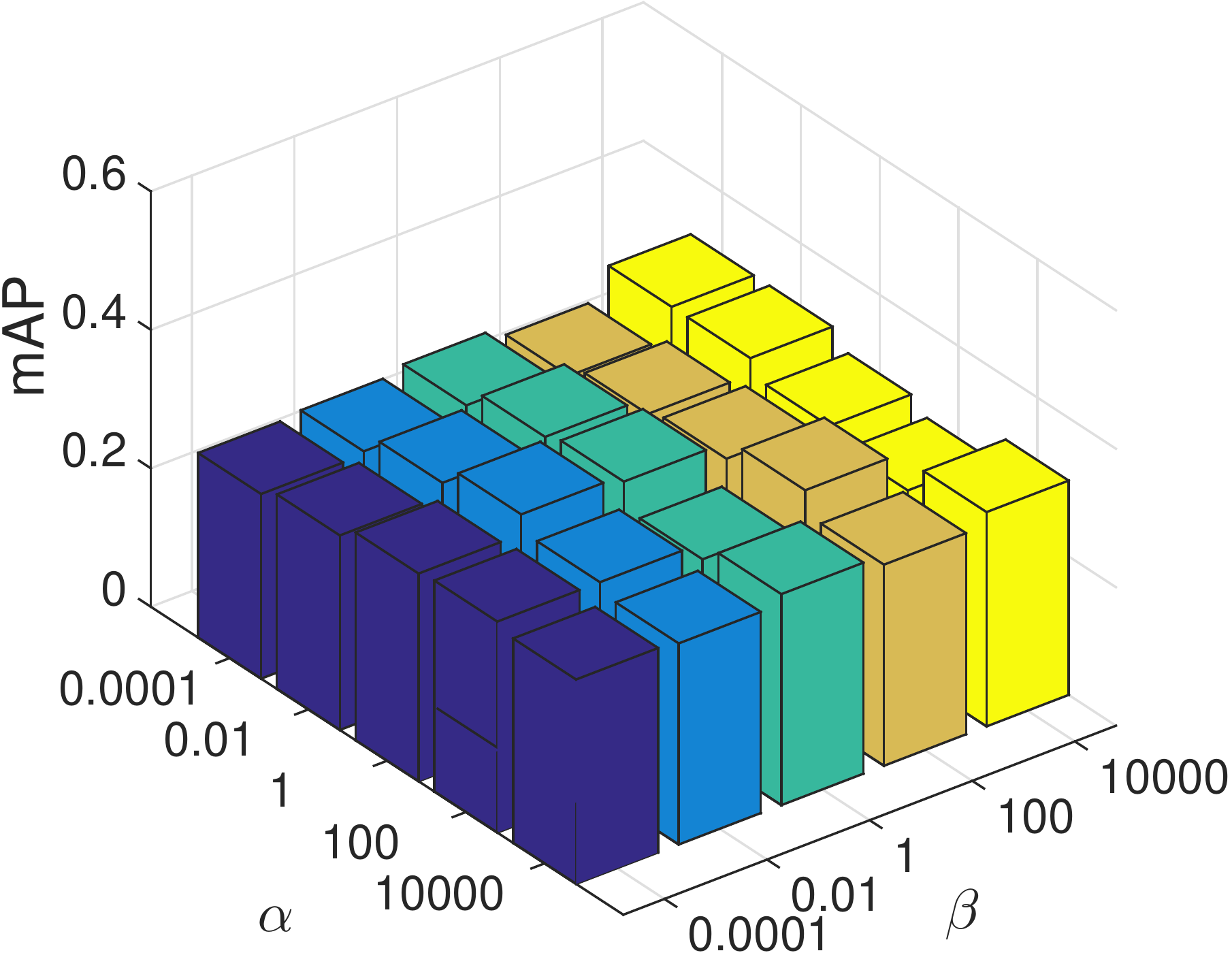}}
\caption{Clustering accuracy variations with parameters $\alpha,\beta,\gamma$ in Eq.(\ref{eq:overall}) on MSRC-V1.}
\label{robustness}
\end{figure*}
\begin{figure*}
\centering                     
\subfigure[MSRC-V1]{
\includegraphics[width=68mm]{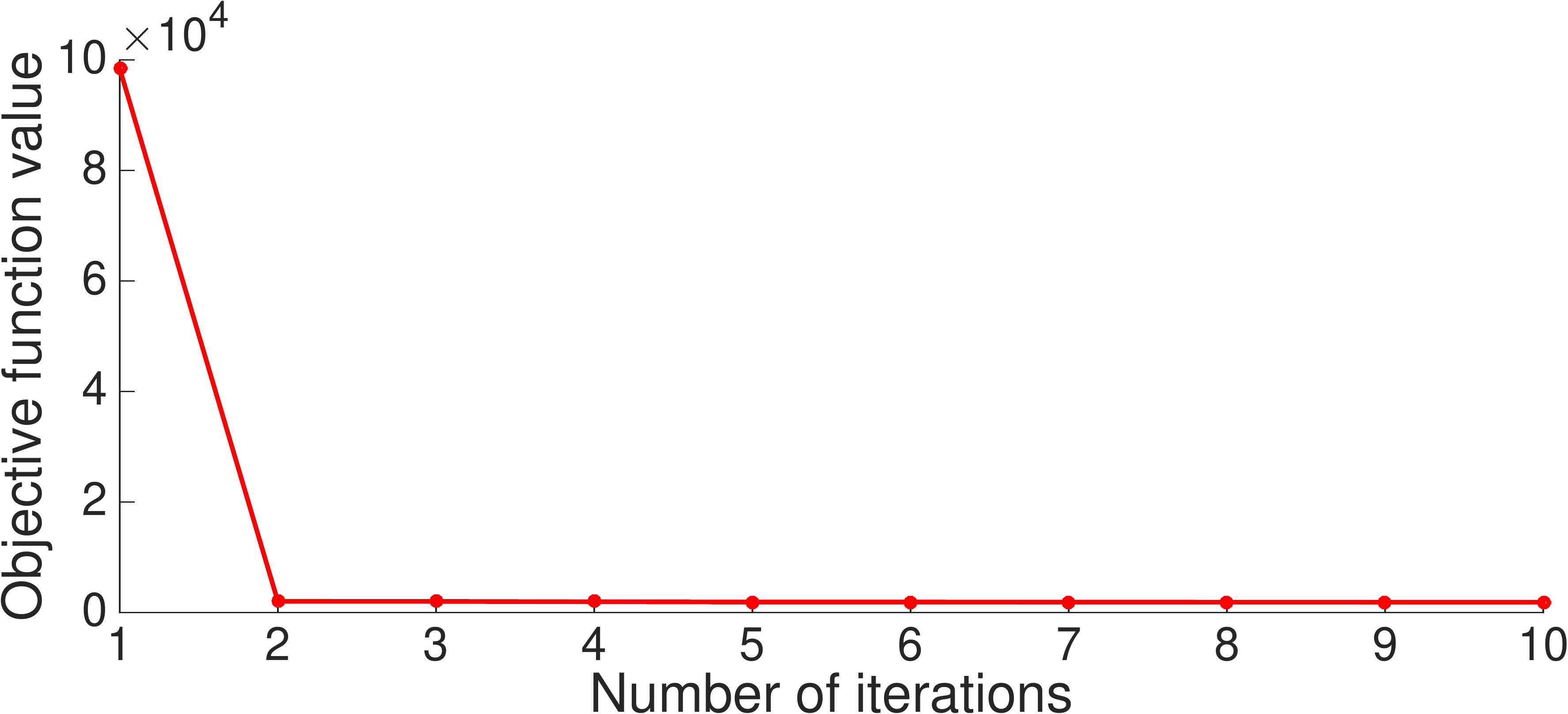}}
\subfigure[Handwritten Numeral]{
\includegraphics[width=68mm]{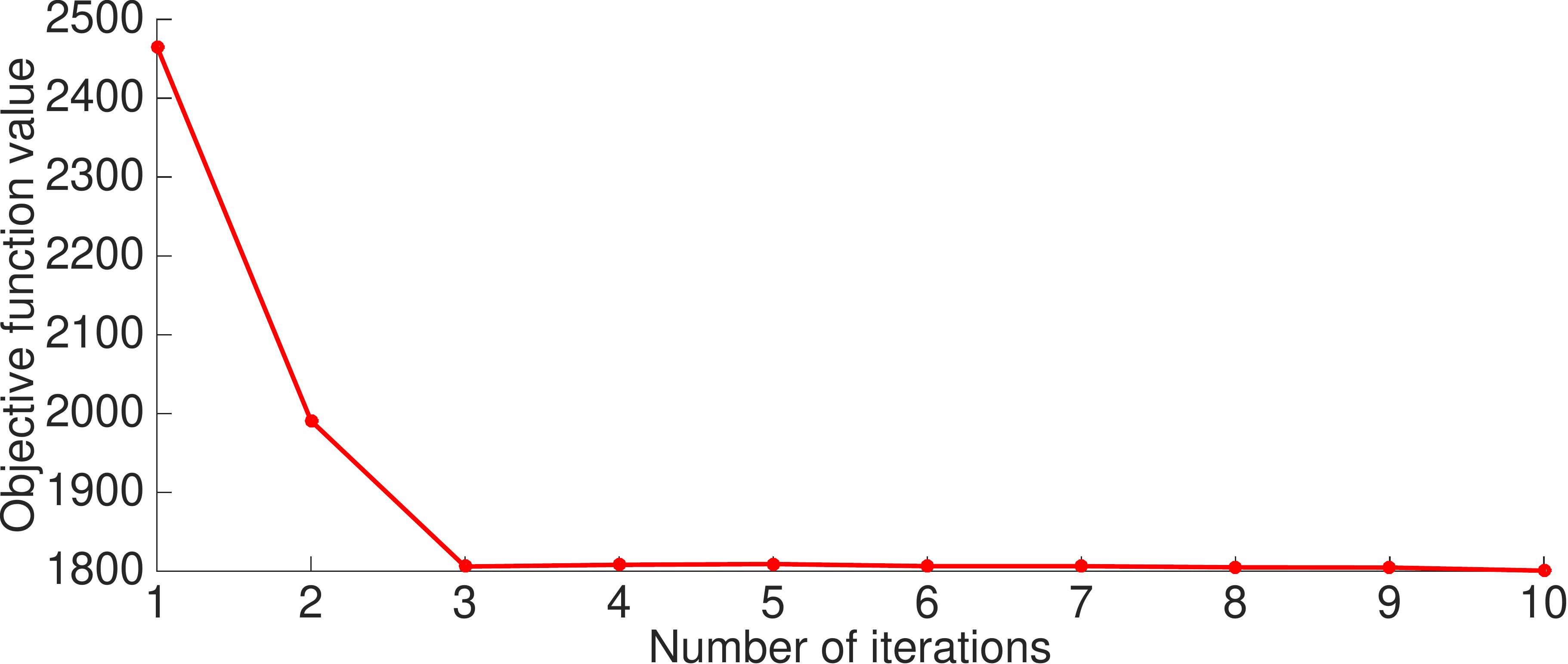}}
\caption{Variations of the objective function value in Eq.(\ref{eq:overall}) with the number of iterations on MSRC-V1 and Handwritten Numeral.}
\label{conver}
\end{figure*}
\subsection{Comparison Results}
The comparison results measured by ACC and NMI are reported in Table \ref{Tab1} and Table \ref{Tab2}, respectively. For these metrics, the higher value indicates the better feature selection performance. Each metric penalizes or favors different properties in feature selection. Hence, we report results on these diverse measures to perform a comprehensive evaluation. The obtained results demonstrate that ACSL can achieve superior or at least comparable performance than the compared approaches. The promising performance of ACSL is attributed to the reason that the proposed collaborative similarity structure learning with proper neighbor assignment could positively facilitate the ultimate multi-view feature selection.



\subsection{Parameter and Convergence Experiment}
We investigate the impact of parameters $\alpha$, $\beta$ and $\gamma$ in Eq.(\ref{eq:overall}) on the performance of ACSL. Specifically, we vary one parameter by fixing the others. Figure \ref{robustness} presents the main results on \textbf{MSRC-V1}.  The obtained results clearly show that ACSL is robust to the involved three parameters.
Figure \ref{conver} records the variations of the objective function value in Eq.(\ref{eq:overall}) with the number of iterations on \textbf{MSRC-V1} and \textbf{Handwritten Numeral}. We can easily observe that the convergence curves become stable within about 5 iterations. The fast convergence ensures the optimization efficiency of ACSL.

\section{Conclusion}
\label{sec:conclu}
In this paper, we propose an adaptive collaborative similarity structure learning for multi-view feature selection. Different from existing approaches, we integrate collaborative similarity learning and feature selection into a unified framework. The collaborative similarity structure with the ideal neighbor assignment and similarity combination weights are adaptively learned to positively facilitate the subsequent feature selection. Simultaneously, the feature selection can supervise the similarity learning process to dynamically construct the desirable similarity structure. Experiments show the superiority of the proposed approach.


\small
\bibliographystyle{named}
\bibliography{ijcai18}

\end{document}